\documentclass[english,onecolumn]{IEEEtran}
\usepackage[T1]{fontenc}
\usepackage[latin9]{inputenc}
\usepackage{amsthm}
\usepackage{graphicx}
\usepackage{amssymb} 

\usepackage{amsmath}
\usepackage{cite}
\usepackage{babel}

\usepackage{subcaption}

\makeatletter
\theoremstyle{plain}
\newtheorem{thm}{\protect\theoremname}
\theoremstyle{definition}
\newtheorem{defn}{\protect\definitionname}
\theoremstyle{plain}
\newtheorem{lem}{\protect\lemmaname}
\theoremstyle{definition}
\newtheorem{example}{\protect\examplename}
\theoremstyle{plain}
\newtheorem{cor}{\protect\corollaryname}
\theoremstyle{plain}

\theoremstyle{remark}

\makeatother

\usepackage{babel}
\providecommand{\corollaryname}{Corollary}
\providecommand{\definitionname}{Definition}
\providecommand{\examplename}{Example}
\providecommand{\factname}{Fact}
\providecommand{\lemmaname}{Lemma}
\providecommand{\theoremname}{Theorem}

\DeclareMathOperator{\rank}{rank}

\begin{document}

\title{On Scaling Rules for Energy of VLSI Polar Encoders and Decoders\thanks{Submitted in part for presentation at the 2016 IEEE International Symposium on Information Theory}}
\author{Christopher G. Blake and Frank R. Kschischang\\Department of
	Electrical \& Computer Engineering\\University of Toronto\\
	\texttt{\small christopher.blake@mail.utoronto.ca} \texttt{\small frank@comm.utoronto.ca}}
\maketitle

\begin{abstract}
It is shown that all polar encoding schemes of rate $R>\frac{1}{2}$
of block length $N$ implemented according to the Thompson VLSI model
must take energy $E\ge\Omega\left(N^{3/2}\right)$. This lower bound
is achievable up to polylogarithmic factors using a mesh network topology
defined by Thompson and the encoding algorithm defined by Arikan.
A general class of circuits that compute successive cancellation decoding  adapted from Arikan's butterfly network algorithm is defined.  It is shown that such decoders implemented on a rectangle grid for codes of rate $R>2/3$
must take energy $E\ge\Omega(N^{3/2})$, and this can also be reached
up to polylogarithmic factors using a mesh network. Capacity approaching sequences of energy optimal polar encoders and decoders, as a function of reciprocal gap to capacity $\chi = (1-R/C)^{-1}$, have energy that scales as
$\Omega\left(\chi^{5.325}\right)\le E \le O\left(\chi^{7.05}\log^{4}\left(\chi\right)\right)$.
\end{abstract}

\section{Introduction}
\IEEEPARstart {C}{ontinuing} recent work on the energy of VLSI error control circuits
\cite{GroverInfoFriction,groverFundamental,blakeConstantPePaper,BlakeLDPCAlmostSure,blakeKschischangEnergyLatencyReliabilityTradeoffs, ganesanPaper},
in this paper we provide lower and upper bound scaling rules
on the energy of VLSI implementations of polar encoders and decoders.
In particular, we show that all polar encoders of constant switching activity factor for codes of rate $R>1/2$
have energy that scales at least as $\Omega\left(N^{3/2}\right)$.
 We describe a class of circuits based on the polar decoding algorithm suggested by Arikan
in \cite{ArikanPolar}. We show that circuits of this type for polar codes of rate $R>2/3$ must take at least $\Omega(N^{3/2})$
energy when its output nodes are arranged on a rectangular
grid. The mesh network topology can also reach this lower bound up
to polylogarithmic factors by using circuits with switching activity factor that decreases with increasing block length.

In a companion paper \cite{blakeKschischangEnergyLatencyReliabilityTradeoffs}, we analyze the tradeoffs in the number of clock cycle, energy, and block error probability for encoder and decoder circuits. For sequences of circuits with variable switching activity factor, Grover \cite{GroverInfoFriction} showed similar scaling rules for the energy of encoders and decoders as a function of block error probability. In particular, these results show that coding schemes with block error probability that scales exponentially in block length $N$ must have energy that scales as $\Omega(N^{3/2})$, and our paper \cite{blakeKschischangEnergyLatencyReliabilityTradeoffs} shows that schemes with constant switching activity factor that reach this lower bound must have number of clock cycles $T\ge \Omega (\sqrt N)$. There exist generalized polar decoders with asymptotic block error probability that scale as $\Theta(e^{-N^{1-\epsilon}})$ for any $\epsilon>0$ \cite{HassaniAlishahiUrbanke} (that is, very close to $O(e^{-N})$), and in this paper we discuss how the energy of polar decoders for such codes implemented on a mesh network can get very close to the $\Omega (N^{3/2})$ energy lower bound implied by Grover \cite{GroverInfoFriction}, implying that the energy lower bound is close to tight. However, this requires decoders with switching activity factor that scales as $\Theta(1/N)$, and number of clock cycles that scales close to $\Theta(N^{3/2})$, in contrast to the clock cycle lower bound of \cite{blakeKschischangEnergyLatencyReliabilityTradeoffs}. This is because of the difficulty of parallelization of the successive cancellation decoding algorithm.

In Section~\ref{sec:priorRelated} we discuss how the results of this paper build upon prior work. In Section~\ref{sec:Computational-Model} we discuss the computational
model we will use and present some basic definitions that will be
used. In Section \ref{sec:Polar-Encoders-Lower} we present one of
the main technical results of this paper, showing a lower bound on
the VLSI energy complexity of polar encoding. We discover a similar
lower bound for the complexity of decoding using VLSI circuits derived
from Arikan's successive cancellation decoding algorithm in Section \ref{sec:Arikan-Successive-Cancellation}.
Upper bounds that reach the lower bounds of the previous two sections
are presented in Section \ref{sec:Upper-Bounds-mesh-network} where
the mesh network used by Thompson for sorting and fast Fourier transform
is applied to the polar encoding and decoding algorithms. In Section~\ref{sec:generalizedOnMesh} we study how some of these results can be extended to polar coding with more general generating matrices. In Section
\ref{sec:Energy-Scaling-as-function-gap-to-capacity} we show how
these upper and lower bound results, when combined with finite length
analysis of polar coding, results in upper and lower bounds for the
energy of polar decoding as a function of gap to capacity. Finally,
in Section \ref{sec:Future-Work} we discuss some open questions and
areas of future work. 

\emph{Notation:} We use standard Bachmann-Landau
notation in our discussions. The statement $f(x)=O(g(x))$ means that for
sufficiently large $x$, $f(x)\le cg(x)$ for some positive constant
$c$. The statement $f(x)=\Omega(g(x))$ means that for sufficiently
large $x$, $f(x)\ge cg(x)$ again for some positive constant $c$.
The statement $f(x)=\Theta(g(x))$ means that there are two positive
constants $b$ and $c$ such that $b\le c$ and for sufficiently large
$x$, $bg(x)\le f(x)\le cg(x)$. The statement $f(x)=o(g(x))$ means that for sufficiently large $x$ and all $c>0$, $f(x) < cg(x)$.

We let the symbol $[N]=\{1,2,\ldots , N\}$ denote the set of integers from $1$ to $N$.

 Given a set of indices
$X,Y\subseteq [N]$, and a vector $V$ of length $N$, we define
the notation $V(X)$ to be the subvector of $V$ formed by the indices in
$X$. As well, given an $N\times N$ matrix $A$, the notation $A(X,Y)$
refers to the submatrix formed by the rows with indices in $X$ and
columns with indices in $Y$. The notation $A(X)$ refers to the submatrix of $A$ formed by the rows with indices in $X$ and all the columns.

\section{Prior Related Work}\label{sec:priorRelated}

This paper analyzes the VLSI complexity of polar encoders and decoders. It was recently discovered that the general technique of polar coding was first discovered by Stolte in \cite{stoltePhDThesis}, though these results were never published and the author did not conjecture that this construction reaches capacity. Arikan \cite{ArikanPolar} independently discovered this technique and proved that such codes can reach capacity. Our work in Sections \ref{sec:Arikan-Successive-Cancellation} and \ref{sec:Upper-Bounds-mesh-network} take inspiration from polar encoding and decoding graphs presented in the Arikan paper. 

Our work in Section \ref{sec:Polar-Encoders-Lower} involves a lower bound for circuits that compute polar-encoding functions. The lower bounding technique comes from Thompson \cite{ThompsonThesis}. The key lemma needed is Lemma~\ref{lem:All-rectangle-pairs}, which is analogous to a property of the discrete Fourier transform (DFT) matrix proved by Valiant in \cite[Lemma 4]{valiant76} and by Tompa in \cite[Lemma 3]{tompaTimeSpaceTradeoffs}, though we use a different technique to derive this property.

In Section~\ref{sec:Arikan-Successive-Cancellation} we study the butterfly network graph proposed by Arikan \cite{ArikanPolar} for polar decoding. Our key lemma shows that the minimum bisection width of this graph's output nodes is $N$. This result is similar to the result of \cite{BornsteinBisectionButterfly} which shows that the minimum bisection width of \emph{all} the nodes the butterfly network graph is $2(\sqrt{2}-1)N+o(N)\approx 0.82 N$. Because of our circuit lower bounding technique, the minimum bisection width of the output nodes is required, and not all the nodes of the graph, motivating our approach.

In Section~\ref{sec:Upper-Bounds-mesh-network} we show how a mesh network can achieve our encoding and decoding energy lower bounds up to polylogarithmic factors. A mesh network DFT algorithm was proposed by Stevens \cite{StevensFourierSubroutine} and shown by Thompson \cite{ThompsonThesis} to reach the DFT VLSI complexity lower bounds.

There have been a number of papers on practical VLSI implementations of polar encoders and decoders \cite{DizdarArikanPolarCodeImplementation, BalatsoukasRaymondGrossBurg,AlamdarKschischang,lerouxTalVardyGrossPolar}, though a theoretical analysis of how the energy of such circuits scale has not been performed. However, these results show that practical polar coding circuits compete well with other error control codes, motivating our theoretical analysis.

\section{Computational Model\label{sec:Computational-Model}}

The model we will use in this paper, which we call the Thompson VLSI model, is adapted from Thompson in \cite{ThompsonThesis}. The precise
model we will use is defined in \cite{blakeConstantPePaper}, with a few minor differences. In the previous model, lower bounds were presented in terms of wire width and a technology constant. In this paper, since we are concerned with scaling rules, to avoid unnecessary notation, we will let these constants be $1$. As well, in this paper, we let the symbol $q$ represent the switching activity factor (that is, the fraction of the circuit that is active on average per clock cycle). In the previous paper $q$ was implicitly $1$. The key circuit parameters we will use are $A$, the circuit area, and $T$, the number of clock cycles. Then, the circuit energy is given by $E=qAT$.
 
  Note that in this model circuit nodes can be considered
the vertices of a graph, and in a circuit's graph two edges connect
a vertex if and only if there is a wire between the two corresponding
nodes. 

\begin{defn}
A\emph{ bisection} of a set of vertices $X$ of a graph $G=\left(V,E\right)$
is a set of edges such that their removal bisects the vertices $X$;
that is, they divide the graph into two disconnected components with
vertices $V_{1}$ and $V_{2}$ such that $\left|\left|X\cap V_{1}\right|-\left|X\cap V_{2}\right|\right|\le1$.
The \emph{minimum bisection width} (MBW) of a set of vertices $X$ is the
size of the bisection of these vertices that is minimal over all such
bisections.
\end{defn}
An important lemma relates the minimum bisection width of a graph
to its circuit area:
\begin{lem}[Thompson]\label{lem:ThompsonLemma} All circuits whose
corresponding graph has MBW $\omega$ have area
bounded by:
\[
A\ge\frac{\omega^{2}}{4}.
\]
\end{lem}
\begin{IEEEproof}
See \cite{ThompsonThesis}.\end{IEEEproof}
\begin{defn}
A \emph{coding scheme} is a sequence of error control codes of increasing
block length $N$, together with a sequence of \emph{encoding circuits}
and \emph{decoding circuits.} It is associated with a particular channel.
\end{defn}
Let $P_{\mathrm{e}}(N)$ be the block error probability for the code in the scheme with block length $N$. Then we can define the following:
\begin{defn} \cite{blakeKschischangEnergyLatencyReliabilityTradeoffs}
An $f(N)$-coding scheme is a coding scheme in which for sufficiently
large $N$, $P_{\mathrm{e}}(N)\le f(N)$.

Note that this definition classifies coding schemes in terms of their
block error probabilities. A ``good'' coding scheme should thus
have this $f(N)$ scale quickly to $0$. Using the result of \cite{ArikanTeletareRateOfPolarization}
we can see that for every $\epsilon>0$, there exists an $e^{-N^{1/2-\epsilon}}$-
coding schemes using polar codes.
\end{defn}

\section{Polar Encoders Lower Bound\label{sec:Polar-Encoders-Lower}}

In this section we will prove that all polar encoders of rate greater
than $1/2$ must have energy that scales as $\Omega\left(N^{3/2}\right)$.
The main technical result will be Lemma \ref{lem:All-rectangle-pairs},
in which we show a property about the rank of rectangle
pairs of the polar encoding generator matrix.

\subsection{Rectangle Pairs}
We will consider an $N\times N$ matrix $G$. We let $R,C\subseteq[N] $.
\begin{defn}
Let $G(R,C)$ be the submatrix of $G$ formed by selecting the rows
with indices in $R$ and the columns with indices in
$C$. We call such an object a \emph{rectangle} of $G$.\end{defn}
\begin{example}
\label{exa:RectangleExample}Let
\[
G=\left[\begin{array}{cccc}
1 & 2 & 3 & 4\\
5 & 6 & 7 & 8\\
9 & 10 & 11 & 12\\
13 & 14 & 15 & 16
\end{array}\right],R=\left\{ 1,3\right\} ,C=\left\{ 2,4\right\} 
\]
 then 
\[
G\left(R,C\right)=\left[\begin{array}{cc}
2 & 4\\
10 & 12
\end{array}\right].
\]
\end{example}
\begin{defn}
If $A\subseteq [N] $, we let the \emph{relative complement} of $A$ be  
$\bar{A}=[N] \backslash A$, that is, those
elements in $\left\{ 1,2,\ldots N\right\} $ that are not in $A$.
The relevant value of $N$ will depend on context, and when in doubt
will be specified clearly.
\end{defn}
Again, let $R,C\subseteq\left\{ 1,2,\ldots,N\right\} $ and $G$ be
an $N\times N$ matrix.
\begin{defn}
A\emph{ rectangle pair} is an ordered pair of submatrices of a given
matrix $G$, associated with two sets $R$ and $C,$ defined as $\left(G(R,C),G(\bar{R},\bar{C})\right)$.\end{defn}
\begin{example}\label{exa:rectanglePair}
If $G,$$R,$ and $C$ are defined as in Example~\ref{exa:RectangleExample},
then the rectangle pair associated with $R$ and $C$ is:

\[
\left(\left[\begin{array}{cc}
2 & 4\\
10 & 12
\end{array}\right],\left[\begin{array}{cc}
5 & 7\\
13 & 15
\end{array}\right]\right)
\]

\end{example}

We shall also define:
\begin{defn}
	A \emph{$k$-row-reduced rectangle pair} of a matrix $G$ is an ordered
	pair of matrices $(X,Y)$. It is formed by starting with any rectangle
	pair $(A,B)$ of $G$ and deleting $a$ rows from $A$ to form
	$X$ and $b$ rows from $B$ to form $Y$ such that $a+b=k$.\end{defn}
\begin{example}
	A $1$-row-reduced rectangle pair of the matrix $G$ from Example~\ref{exa:rectanglePair} is 
	\[
	\left(\left[\begin{array}{cc}
	10 & 12\end{array}\right],\left[\begin{array}{cc}
	5 & 7\\
	13 & 15
	\end{array}\right]\right).
	\]
	which is formed by deleting a row from the first matrix in the rectangle
	pair
	\[
	\left(\left[\begin{array}{cc}
	2 & 4\\
	10 & 12
	\end{array}\right],\left[\begin{array}{cc}
	5 & 7\\
	13 & 15
	\end{array}\right]\right)
	\]
	of $G$.\end{example}
We will consider the structure of the polar encoding matrix by considering
its rectangle pairs. 

\subsection{Universal Polar Coding Generator Matrix Properties}
\begin{defn}\label{def:recursiveFn}
The \emph{universal polar coding generator matrix}  $G_{n}$ is a matrix defined recursively
by Arikan \cite{ArikanPolar}. Let: 
\[
F_{1}=\left[\begin{array}{cc}
1 & 0\\
1 & 1
\end{array}\right]
\]
and 
\[
F_{n}=\left[\begin{array}{cc}
F_{n-1} & 0\\
F_{n-1} & F_{n-1}
\end{array}\right].
\]
Then 
\[
G_{n}=B_{n}F_{n}
\]
where $B_{n}$ is a permutation matrix interpreted as the bit-reversal
operator. 
\end{defn}
The structure of $B_{n}$ (other than it being a matrix that permutes
the rows of $F_{n}$) will not matter in the proofs that follow.

We will prove a theorem showing that the sum of the ranks (over the
field $\mathbb{F}_{2}$) of the entries of any rectangle pair of $G_{N}$ is
at least $N/2$. This will imply high VLSI complexity for sufficiently
high rate VLSI polar encoders.

We will first consider the ranks of rectangle pairs of $F_{n}$. Note
that $F_{n}$ and $G_{n}$ are $N\times N$ matrices, where $N=2^{n}$.
\begin{lem}
	\label{lem:LemmaForRecursion}Let $X$ be a matrix with entries
	in a field partitioned as
	\[
	X=\left[\begin{array}{c|c}
	A & 0\\\hline
	C & B
	\end{array}\right],
	\]
	where $0$ is a zero submatrix.
	
	Then $\rank(X)\ge \rank(A)+\rank(B)$.\end{lem}
\begin{IEEEproof}
	There are $\rank(A)$ linearly
	independent rows of $A$, and $\rank(B)$ linearly independent rows
	of $B$. The $\rank(A)+\rank(B)$ rows of $X$ corresponding to these rows are also linearly independent.
\end{IEEEproof}

\begin{lem}
\label{lem:All-rectangle-pairs} All rectangle pairs $\left(A_{n},B_{n}\right)$
of $F_{n}$ have $\rank(A_{n})+\rank(B_{n})\ge\frac{N}{2}$.
\end{lem}
\begin{IEEEproof}
We will use mathematical induction. Note that \[F_{1}=\left[\begin{array}{cc}
1 & 0\\
1 & 1
\end{array}\right],\]
and checking all the possible rectangle pairs $\left(A_{1},B_{1}\right)$
of $F_{1}$ we have that $\rank(A_{1})+\rank(B_{1})\ge1=\frac{N}{2}$.
Now, we assume that for all $k\le n-1$, for all rectangle pairs of
$F_{k}$ denoted $(A_{k},B_{k})$:
\begin{equation}
\rank(A_{k})+\rank(B_{k})\ge\frac{2^{k}}{2}.\label{eq:inductionHypothesis}
\end{equation}

Consider a rectangle pair $\left(A_{n},B_{n}\right)$ of $F_n$. Note that $A_{n}$ can be written as:
\begin{equation}
A_{n}=\left[\begin{array}{cc}
P & 0\\
Q & S
\end{array}\right].\label{eq:AnNotCompact}
\end{equation}
and $B_n$ can be written as:
\[
B_{n}=\left[\begin{array}{cc}
L & 0\\
M & J
\end{array}\right].
\]
where $(P, L)$ and $(S, J)$ are rectangle pairs of $F_{n-1}$.

Observe that, by Lemma~\ref{lem:LemmaForRecursion} 
\begin{equation}
\rank(A_{n})\ge \rank(P)+\rank(S)\label{eq:rankOfAn}
\end{equation}
and
\begin{equation}
\rank(B_{n})\ge \rank(L)+\rank(J).\label{eq:rankOfBn}
\end{equation}
Since $(P,L)$ and $(S,J)$ are rectangle pairs of $F_{n-1}$, by the induction hypothesis (\ref{eq:inductionHypothesis}):
\begin{equation}
\rank(P)+\rank(L)\ge\frac{N}{4}\label{eq:rankLplusrankP}
\end{equation}
and 
\begin{equation}
\rank(S)+\rank(J)\ge\frac{N}{4}.\label{eq:rankJplusrankT}
\end{equation}
Thus, by combining (\ref{eq:rankOfAn}) and (\ref{eq:rankOfBn}):
\begin{eqnarray*}
	\rank(A_{n})+\rank(B_{n}) & \ge\\
	\rank(P)+\rank(S)+\rank(L)+\rank(J)
\end{eqnarray*}
By rearranging the right side of this inequality and directly substituting the bounds of (\ref{eq:rankLplusrankP}) and (\ref{eq:rankJplusrankT}) we get:
\[
\rank(A_{n})+\rank(B_{n})\ge\frac{N}{4}+\frac{N}{4}=\frac{N}{2}.
\]
\end{IEEEproof}
\begin{cor}
\label{cor:sumOfRanksOfGn}For any rectangle pair of $G_{n}$ denoted
$(A,B)$:
\[
\rank(A)+\rank(B)\ge\frac{N}{2}.
\]
\end{cor}
\begin{IEEEproof}
This Corollary follows by observing that \cite{ArikanPolar} $G_{n}=B_{n}F_{n}$
where $B_{n}$ is a permutation matrix that permutes the rows of $F_{n}$.
For every rectangle pair of $F_{n}$ there is an equivalent rectangle
pair of $G_{n}$, selected by choosing the same columns and choosing
the rows as permuted by $B_{n}$. The rectangles forming such a rectangle
pair will have the same rows, simply permuted. Thus they have the
same row space and thus the same rank. 
\end{IEEEproof}

\subsection{Encoder Circuit Lower Bounds}
We consider below a circuit that computes a polar encoding function.
Such a function is associated with a set $\mathcal{A}$ of free indices. We
denote the vector of free indices $u({\mathcal{A}})\in\{0,1\}^{\left|\mathcal{A}\right|}$.
It is also associated with a vector of frozen indices $u({\bar{\mathcal{A}}})\in\{0,1\}^{N-\left|\mathcal{A}\right|}$.
\begin{defn}\cite{ArikanPolar}
A\emph{ polar encoding function} $f:\{0,1\}^{\left|\mathcal{A}\right|}\rightarrow\{0,1\}^{N}$
associated with free indices $\mathcal{A}$ and frozen vector $u(\bar{\mathcal{A}})$
is defined as
\[
f(u(\mathcal{A}))=u(\mathcal{A})G_{n}(\mathcal{A})+u(\bar{\mathcal{A}})G_{n}(\bar{\mathcal{A}})
\]
where $G_{n}(\mathcal{A})$ is the submatrix of $G_{n}$ formed by
the rows with indices in $\mathcal{A}$, and addition is performed in $\mathbb{F}_2^N$. Such a function is an encoding
function for a code with block length $N=2^{n}$ and rate $R=\frac{\left|\mathcal{A}\right|}{N}$.
\end{defn}

\begin{thm}\label{encodingLowerBoundTheorem}
The area $A$ and the number of clock cycles $T$ for a circuit that
computes a polar encoding function of rate $R$ greater than $\frac{1}{2}$
has area $A$ and number of clock cycles $T$ bounded by

\begin{equation}
AT^{2}\ge\frac{N^{2}\left(1-2R\right)^{2}}{64}=\Omega\left(N^{2}\right)\label{eq:ATsquaredBound}
\end{equation}
and, if such a circuit has switching activity factor $q$, its energy is bounded by
\begin{equation}
E\ge q\frac{N^{3/2}(1-2R)}{8}=\Omega\left(N^{3/2}\right).\label{eq:ATbound}
\end{equation}
\end{thm}
\begin{IEEEproof}
We will follow a similar line of reasoning used by Thompson in analyzing
the complexity of Fourier transform circuits \cite{Thompson}. There
are $N$ output bits of the encoder. We label the indices
of the output nodes on one side of the bisection $L$ (the left side),
and the others $R$ (the right side). The subcircuit containing the
output nodes $L$ will have some inputs bits, labelled
$L_{i}$. Similarly, label the input bits on the right side $R_{i}$.
We denote the vector of inputs on the left side $u(L_i)$, on the right side $u(R_i)$, and the frozen vector $u(\bar {\mathcal A})$.
By simply expanding the vector matrix multiplication, we see that
the left side of the circuit must compute the vector:

\begin{eqnarray*}
	y(L) & = & u(L_{i})G_{n}(L_{i},L)+u(R_{i})G_{n}\left(R_{i},L\right)\\
	&  & +u(\bar{\mathcal{A}})G_{n}\left(\bar{\mathcal{A}},L\right)
\end{eqnarray*}
and similarly the right side must compute the vector:
\begin{eqnarray*}
	y(R) & = & u(L_{i})G_{n}(L_{i},R)+u(R_{i})G_{n}\left(R_{i},R\right)\\
	&  & +u(\bar{\mathcal{A}})G_{n}\left(\bar{\mathcal{A}},R\right).
\end{eqnarray*}
The subcircuits must compute these values only with the input bits
injected into their own input nodes and the bits communicated to them
from the other subcircuit (which of course has access to the other input
nodes). Note that
$\left(G_{n}\left(R_{i},L\right),G_{n}(L_{i},R)\right)$ is an
$\left|\bar{\mathcal{A}}\right|$-row-reduced rectangle pair of $G_{n}$.
Observe from Corollary~\ref{cor:sumOfRanksOfGn} that the sum of the
ranks of these matrices must be at least
$\frac{N}{2}-\left|\bar{\mathcal{A}}\right|=\frac{N}{2}-(1-R)N$, which is greater than $0$ because $R>1/2$.  Thus,
at least $\frac{N}{2}-(1-R)N$ bits must be communicated across this
bisection during the computation. If the circuit has MBW of output bits $\omega$, since
at each clock cycle at most $2\omega$ bits can be communicated across
the bisection:
\begin{equation}
T\ge\frac{\frac{N}{2}-(1-R)N}{2\omega}.\label{eq:TboundInTermsOfOmega}
\end{equation}
By Lemma~\ref{lem:ThompsonLemma}, we have $A\ge\frac{\omega^{2}}{4}$
and thus, combining this with (\ref{eq:TboundInTermsOfOmega}) 
implies
\begin{equation}
AT^{2}\ge\frac{N^{2}\left(2R-1\right)^{2}}{64}.\label{eq:ATsquaredbound}
\end{equation}
Also note that
\[
A\ge N
\]
and thus, combining this inequality with (\ref{eq:ATsquaredbound}) and taking the square root we get
\[
E=qAT\ge q\frac{N^{1.5}(2R-1)}{8}.
\]

\end{IEEEproof}

\section{Arikan Successive Cancellation Polar Decoding Scheme\label{sec:Arikan-Successive-Cancellation}}

In Arikan's original paper on polar coding \cite{ArikanPolar}, the
author presented a Turing-time complexity $O(N\log N)$ algorithm
for computing successive cancellation decoding of polar codes. In
this section, we provide a definition of a polar decoder based on
Arikan's \cite{ArikanPolar} paper, and show that such circuits, when
implemented with output nodes arranged in a rectangular grid, take
energy at least $\Omega(N^{3/2})$.

\subsection{Decoding Complexity Lower Bound}

Below we consider a generalization of the minimum bisection width
of a set of vertices, where instead of dividing the set of vertices
into two sets of equal size, instead we divide the vertices of a set
into two sets, where the size of one of these sets is fixed. 
\begin{defn}
Given a graph $G=\left(V,E\right)$, an $m$\emph{-partition} of a
set of vertices $X\subseteq V$ is an ordered pair $(A,V\backslash A)$
in which $A\subseteq V$ and $\left|A\cap X\right|=m$. The \emph{width}
of this partition is the size of the set of edges connecting the vertices
in $A$ with $V\backslash A$. The $m$\emph{-section width} of a
set $X$ of vertices of a graph $G$ is the minimum width over all
the graph's $m$-partitions of $X$. 

In Figure~\ref{fig:mSectionWidth} we give an simple example of a $2$-partition of a
subset of edges of a graph.

\begin{figure} 
	\centering
	\includegraphics{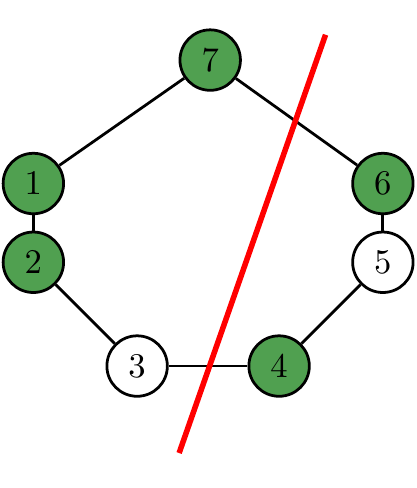}
	\caption{Example of a graph with a $2$-partition of the shaded nodes.
		The labelled partitioning line is also a $1$-partition of the white
		nodes. Inspection shows that such $m$-partitions are minimal. Therefore,
		the $2$-section width of the shaded nodes in this graph is $2$,
		as is the $1$-section width of the white nodes. } \label{fig:mSectionWidth}
\end{figure}

\end{defn}
Note that if a graph $G=(V,E)$ has $2m$ vertices, then the $m$-section
width of $V$ is the same as the graph's minimum bisection width.
\begin{defn}
An $n$\emph{-polar decoding graph,} denoted $P_{n}$, is a generalization
of the graph presented by Arikan in \cite{ArikanPolar} describing
the communication graph of a polar decoding algorithm. It is defined
recursively in Figures \ref{fig:decoderGraphBase} and \ref{fig:decoderGraph}. For the base case, the $1$-polar decoding
graph is the bowtie shaped graph given in Figure~\ref{fig:decoderGraphBase}. An $n$-polar
decoding graph consists of $2^{n}$ nodes called \emph{symbol nodes}
connected to two copies of the $\left(n-1\right)$-polar decoding
graphs as shown in Figure~\ref{fig:decoderGraph}.
\end{defn}
\begin{figure} 
	\centering
	\includegraphics{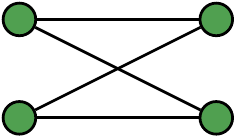}
	\caption{The base case: the decoding graph $P_{1}$.} \label{fig:decoderGraphBase}
\end{figure}

\begin{figure} 
	\centering
	\includegraphics[width=2.8 in]{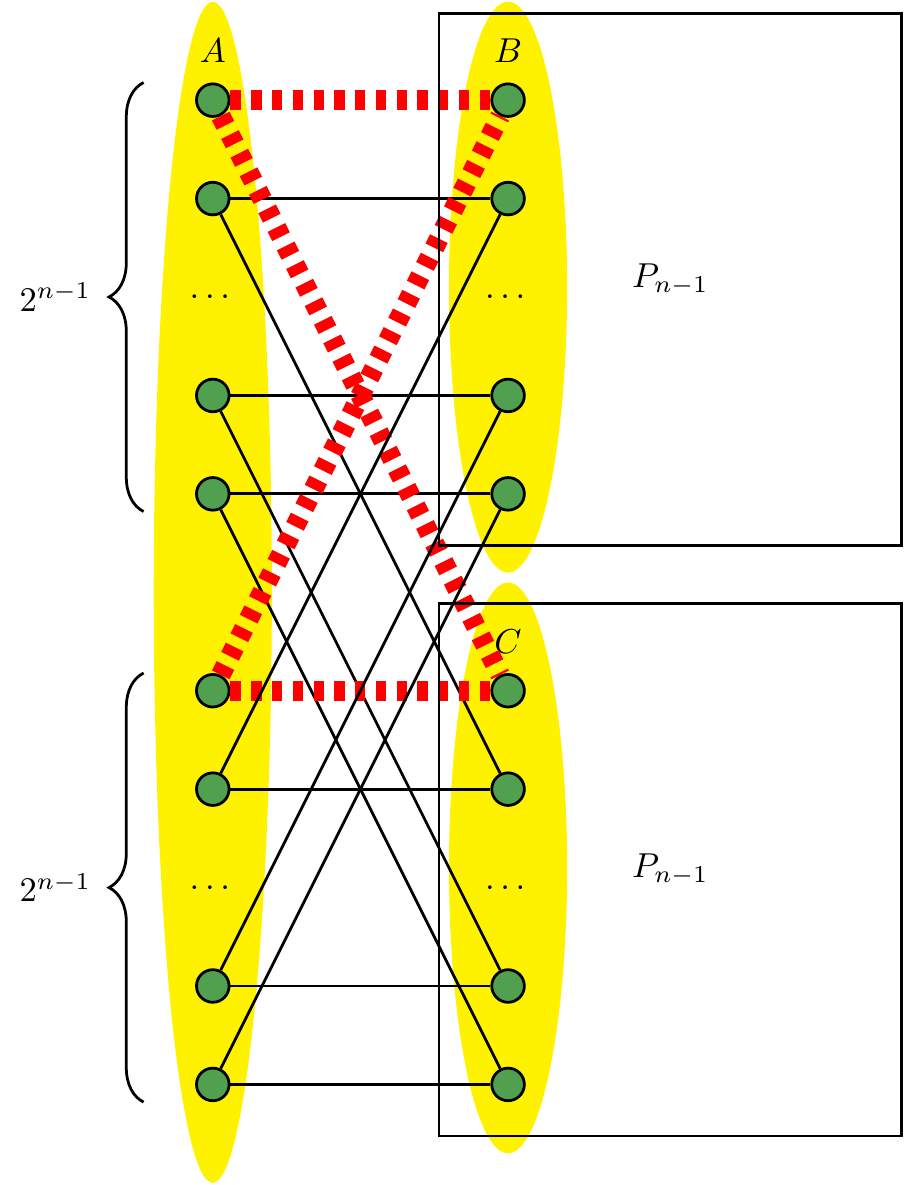}
	
	\caption{A diagram of the recursive structure of a polar decoding
		graph $P_{n}$. The vertices on the left (outlined in a shaded ellipse and labelled $A$) indicate the symbol nodes
		of the graph. The vertices on the right in the two boxes indicate
		the symbol nodes of the two subgraphs (outlined and labelled $B$ and $C$) which are the smaller polar
		decoding graphs $P_{n-1}$.  An example of a
		subgraph that is a bowtie is indicated with edges drawn with a thick,
		dashed line.
		} \label{fig:decoderGraph}
\end{figure}

We shall study the structure of an $m$-partition of a polar decoding
graph.

We call the nodes in the left-most column in the graph of Figure~\ref{fig:decoderGraph} \emph{symbol nodes}. Note that for any $n$-polar decoding graph, there exist $2^{n}$
symbol nodes, as well as symbol nodes of the two $(n-1$)-polar decoding
subgraphs that form the graph. Let the set of symbol nodes of the
larger graph be labelled $A$, the symbol nodes of one of the subgraphs
be $B$, and the other symbol nodes $C$. This labelling of sets is visible in Figure~\ref{fig:decoderGraph}.

Note by inspection that that the bipartite subgraph connecting the
nodes of $A$ with $B\cup C$ consists of \emph{bowties}, subgraphs
containing two vertices from $A$ and a vertex from $B$ and $C$.
An example bowtie is labelled in Figure~\ref{fig:decoderGraph}. A bowtie object
is a subgraph of the polar decoding graph associated with a particular partition. We classify
these bowties according to how the particular partitioning line divides the nodes in the graph.

We now consider a minimum $m$-partition of $A$. Such a partition divides the set of vertices into
two subsets, which we will call the top half and the bottom half.
We can consider how the $2^{n-1}$ bowties connecting $A$ with $B$
and $C$ are split by the $m$-partition. We divide such bowties into
three categories: split bowties, contained bowties, and crossing bowties. 
\begin{defn}
A \emph{split bowtie} is a bowtie in which one element of $A$ is
in the upper half, and one in the lower half. Two examples are given in Fig.~\ref{fig:splitBowties}.
\end{defn}
Note that a split bowtie has $2$ edges crossing the $m$-partition.

\begin{figure}
	
	\begin{subfigure}{0.38\textwidth}
		\includegraphics[width=1.5 in]{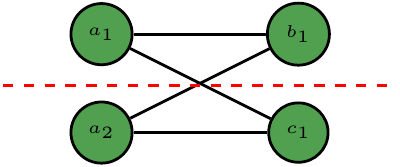}\includegraphics[width=1.5 in]{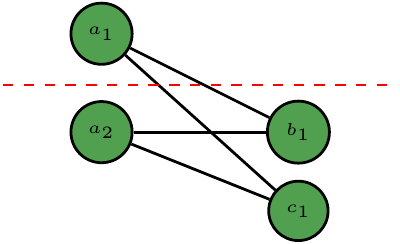}
		\caption{Two examples of split bowties. Split bowties are bowties in which the two nodes in $A$ are on opposite sides of the partition line. It does not matter where the other nodes lie.} \label{fig:splitBowties}
	\end{subfigure}
	\begin{subfigure}{0.38\textwidth}
		\centering

		\includegraphics[width=1.5 in]{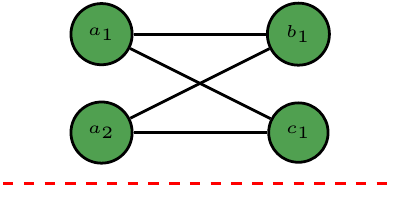}
		\caption{An example of a contained bowtie. Such a bowtie occurs when all the
			nodes of a bowtie are on the same side of the partition. } \label{fig:containedBowtie}
	\end{subfigure}
	\begin{subfigure}{0.38\textwidth}
\centering
		\includegraphics[width=1.5 in]{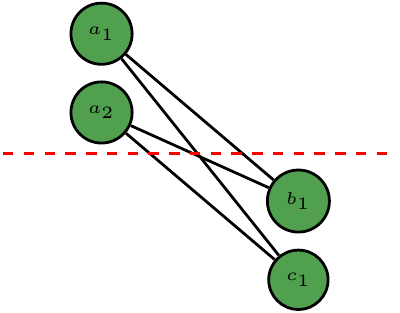}\includegraphics{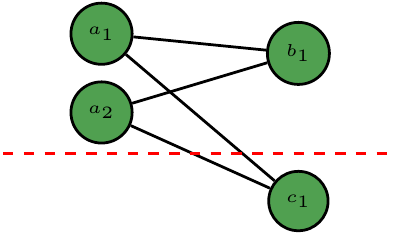}
		\caption{Two examples of crossing bowties. In such bowties the two nodes in $A$ are on the same side of the partition, and at least one of the other nodes in the bowtie is on the other side.} \label{fig:crossingBowties}
	\end{subfigure}
	\caption{Diagrams of split, contained, and crossing bowties. Note that a bowtie is an object associated with a polar decoding graph and a particular partition. In each figure, the nodes labelled $a_1$ and $a_2$ are nodes from the set $A$. Similarly, the nodes labelled $b_1$ are nodes from the set $B$ and those labelled $c_1$ are from the set $C$. The dashed line indicates the relative position of the partitioning line.} \label{fig:1}
\end{figure}

\begin{defn}
A \emph{contained bowtie} is one in which all vertices are either
in the top or bottom half. An example is given in Figure~\ref{fig:containedBowtie}.
\end{defn}
Note that a contained bowtie has no edges crossing the $m$-secting
cut.
\begin{defn}
A \emph{crossing bowtie has two nodes of $A$ on  one side of the
partition line, and at least one of the nodes in $B$ or $C$ on the
other side of the partition line.}
\end{defn}
It should be clear that for any partition of the polar decoding graph, all bowties are either split, contained, or crossing bowties.

We let $G$ be an arbitrary polar decoding graph.
\begin{lem}
\label{lem:OnlySplitandcontainedbowties}Consider an $n$-polar decoder
graph, and let $0\le m\le2^{n-1}$. Such a graph has a minimum $m$-section
partition of the symbol nodes of $G$ which contains only split bowties
and contained bowties.\end{lem}
\begin{IEEEproof}
We provide an exchange argument. For a particular minimal $m$-partition, we argue that nodes in a crossing bowtie can be moved to another side, resulting in a contained bowtie without increasing the number of edges crossing the partition line.  By examining the left side of Figure~\ref{fig:crossingBowties}, there are $4$ edges crossing the partition
line. For crossing bowties of this type, moving vertices $b_1$ and $c_1$ to the side containing the vertices in $A$ decreases the
edges of this bowtie crossing the partitioning line by $4$. The degree
of any vertex in this graph is at most $4$, thus this can, at worst,
result in $4$ new edges crossing the partition (which would result
from the two extra edges each on vertices $b_1$ and $c_1$ now crossing
the partition line). A similar argument can be made for crossing bowties
like those in the right side of Figure~\ref{fig:crossingBowties}. Thus, any minimum $m$-section partition of the output
nodes containing crossing bowties can be modified to one that contains
no crossing bowties.
\end{IEEEproof}
We shall prove a lemma regarding the $m$-section width of an $N$-polar
decoding communication graph. 
\begin{lem}
\label{lem:LemmamSectionWidth}Let $0\le m\le2^{n-1}$. Then the $m$-section
width of the symbol nodes of an $n$-polar decoding graph is at least
$2m$.\end{lem}
\begin{IEEEproof}
We shall prove this by induction. For the base case, the $n=1$ polar
decoding graph given in Figure~\ref{fig:decoderGraphBase} can be shown by inspection to satisfy
the lemma by simply checking the width of all $0$ and $1$-partitions.

We shall assume that the $m$-section width, where $0\le m\le2^{k-2}$
, of a $\left(k-1\right)$-polar decoding graph is $2m$.

We consider a minimal $m$-partition that contains only split and contained bowties that exists by Lemma~\ref{lem:OnlySplitandcontainedbowties}. Denote one half of the partition
the upper half and the other the lower half. Without loss of generality
we assume the upper half contains $m$ nodes of $A$.  Let the number
of contained bowties in the upper half be $C_{\mathrm{upper}}$, and
the number in the lower half $C_{\mathrm{lower}}$. Let the number
of split bowties be $S$. Since the upper half contains $m$ nodes
of $A$, then
\begin{equation}
2C_{\mathrm{upper}}+S=m.\label{eq:mInTermsOfCupperandS}
\end{equation}
Note that the number of contained bowties on the lower half must at least equal the number on the upper half, since the number of symbol nodes on the lower half must equal or exceed the number in the upper half. Thus, there must be at least $m$ elements of $B$ and $m$ elements of $C$ on both side of the partition. As well,   at least one of these sides cannot have more than $N/4$ elements (since each of these sets contains only $N/2$ elements in total). Thus, there is an 
$x$-partition of both $B$ and $C$ induced by the partition, where
$C_{\mathrm{upper}} \le x\le N/4$.  Thus, each of these partitions must have
at least $2C_{\mathrm{upper}}$ edges crossing the partition line
by the induction hypothesis. In addition, there are $2$ edges crossing
the partition line for each split bowtie (easily observed by inspecting
Figure~\ref{fig:splitBowties}). Thus, the number of edges crossing the partition line
is at least
\[
\mathrm{Edges\ crossing}\ge4C_{\mathrm{upper}}+2S=2m
\]
where we have applied (\ref{eq:mInTermsOfCupperandS}), proving the theorem.\end{IEEEproof}

We will consider algorithms whose communication graph is based on the polar decoding graph. However, bits corresponding to certain symbol nodes which are frozen obviously do not need to have their own node in a communication graph. Thus we consider a frozen-bit polar decoding graph.

\begin{defn} An  $n$-\emph{frozen bit polar decoding graph} associated with frozen bit indices $\mathcal{\bar{A}}$ is a graph obtained by deleting the symbol nodes corresponding to   $\mathcal{\bar{A}}$ from $P_n$ and also the edges to which they are connected. The symbol nodes that remain are called the \emph{unfrozen nodes}. Such a graph is a decoding graph for a rate $R=1-\frac{|\mathcal{\bar{A}}|}{N}$ code.
	\end{defn}
	Note that this is a natural simplification of the polar decoding graph once frozen bits are considered. However, this is not the only possible simplification. In this paper we only consider simplifications that involve deletion of the nodes corresponding to the frozen bits.
	
	Once the symbol nodes corresponding to frozen bits are deleted, we then consider the bisection width of the remaining symbol nodes.
	
\begin{cor}\label{cor:unfrozenSymbolNodesMBW} The minimum bisection width $\omega$ of the unfrozen nodes of any $n$-frozen bit polar decoding graph in which $R>2/3$ is at least:
	\[ \omega \ge N(3R-2)
	\]
	\end{cor}
\begin{IEEEproof}
	Suppose not. Then, consider the unfrozen symbol nodes minimal bisection with $W<N(3R-2)$ nodes crossing it. Now, add the frozen symbol nodes and their edges back to this graph. There are at most $(1-R)N$ such nodes to be added, and thus they can increase the number of edges in the graph by at most $2(1-R)N$. At most all these edges can cross the partition line, and thus this partition line can have at most $W+2(1-R)<RN$ edges crossing it. Note that the partition line forms an $m$-partition of all the symbol nodes, where $m\ge RN/2$, But, by Lemma~\ref{lem:LemmamSectionWidth}, any such $m$-partition must have at least $RN$ edges crossing it, resulting in a contradiction.
	\end{IEEEproof}
	
\subsection{Decoder VLSI Lower Bounds}

 Note that a Thompson circuit
is associated with a graph. In the course of a computation, messages
will be passed from node to node in the circuit. Each binary message
passed corresponds to another edge in the computation's \emph{communication
graph}. We will define a polar decoder (a type of circuit) in terms
of a circuit's communication graph.

 We first define contracting vertices
and subdividing edges.
\begin{defn}
Any two vertices of a graph $u$ and $v$  can be \emph{contracted} by joining two vertices
together, resulting in a new vertex $v'$. The vertex $v'$ will have all the edges corresponding to both $u$ and $v$ with the exception of any edges that joined $u$ and $v$ in the original graph, which are deleted.
\end{defn}

\begin{defn}
	An edge of a graph can be \emph{subdivided} by being replaced by two edges with a vertex in the middle.
\end{defn}

Consider a frozen polar decoding graph. Such a graph has $RN$ symbol nodes. Now, consider a process in which vertices and edges are added to this graph, and some edges are subdivided and some vertices are contracted, but no two symbol nodes are contracted.

\begin{defn}
	Such a graph is called a \emph{polar decoding communication graph}. Obviously, such a graph has $RN$ symbol nodes, which are the nodes that correspond to the symbol nodes of the original frozen polar decoding graph. We call these nodes the graph's \emph{output nodes.}
\end{defn}

\begin{lem}
\label{lem:mergingOutputNodes} All polar decoding communication graphs obtained from a rate $R>2/3$ frozen polar decoding graph have minimum bisection width of output nodes at least $N(3R-2)$.
\end{lem} 
\begin{proof}
	First, observe from Corollary~\ref{cor:unfrozenSymbolNodesMBW} that the minimum bisection width of the symbol nodes of the original frozen polar decoding graph is at least $N(3R-2)$. Suppose that the polar decoding communication graph obtained from this graph has minimum bisection width less than this. Then consider the partition with $W<N(3R-2)$ edges crossing the partition line. We can reverse the operations that obtained the the communication graph and this can only decrease the number of edges crossing the partition line, resulting in a bisection of less than $N(3R-2)$ for the original frozen polar decoding graph, a contradiction. 
	To see this, consider what happens when a vertex $v'$ that was contracted has this operation reversed, resulting in two vertices with $u$ and $v$ with (possibly) an edge between them. Place $u$ and $v$ on the same side of the partition previously occupied by $v'$. It is then obvious that any edge between them cannot increase the width of this partition. It is also obvious that deleting any vertices and edges that were added can only decrease the the number of edges crossing the partition line. The same is true for reversing a subdividing process.
\end{proof}

\begin{defn}
	An \emph{Arikan polar decoding circuit} is a circuit whose associated communicaiton graph is also a polar decoding communication graph.
\end{defn}

Note that in our model, a Thompson circuit is created by placing nodes and edges each in a grid of squares of side length $1$.  Consider placing a  grid of (possibly larger) squares with integer side length on top of any Thompson circuit and with sides aligned to the smaller grid of squares defining the VLSI circuit. Then each square in the grid will contain some output nodes. 
\begin{defn}
	A \emph{rectangle grid output circuit} with $N$ output nodes is a circuit in which there is a grid of squares that can be placed upon the circuit, and there is a $\sqrt N \times \sqrt N$ array of larger grid squares, each which contains exactly one output node. 
\end{defn}
\begin{example}
	The mesh network defined in Section~\ref{sec:Upper-Bounds-mesh-network} is an example of a rectangle grid output circuit.
\end{example} 
We suspect that our scaling rule lower bounds for polar decoders would extend to implementations that are not necessarily rectangle grid output circuits, however for simplicity we only present our results for such circuits. A generalization of the following lemma to a broader class of circuits would be sufficient:

\begin{lem}\label{lem:rectangleMBWEnergyLemma}
All rectangle grid output circuits with $\Theta (N)$ output nodes
and with a communication graph with minimum bisection width $\omega$
have energy $E\ge\Omega\left(\omega \sqrt{N}\right)$.\end{lem}
\begin{IEEEproof}
See Appendix~\ref{app:uniformOutputRectangleMBWLemma}.\end{IEEEproof}
Note that the above lemma does not assume a constant switching activity factor.
\begin{thm}\label{decoderLowerBoundTheorem}
All rectangle grid output, Arikan polar decoding circuits
have energy that scales at least as $\Omega\left(N^{3/2}\right)$.\end{thm}
\begin{IEEEproof}
This flows directly from Lemma~\ref{lem:mergingOutputNodes}
and Lemma~\ref{lem:rectangleMBWEnergyLemma}.
\end{IEEEproof}
It is now trivial to observe that this lower bound
can be reached up to a polylogarithmic factor on a mesh network which we discuss in the Section~\ref{sub:Decoding-Mesh-Network}.

\section{Upper Bounds\label{sec:Upper-Bounds-mesh-network}}

\subsection{Mesh Network}

We will show that, up to polylogarithmic factors, the lower bounds
on the energy of polar encoding and decoding complexity can be reached.
The mesh network topology that we present to meet these bounds derives
from Thompson \cite{ThompsonThesis}. 

\begin{figure} 
	\centering
	\includegraphics[width=2.8 in]{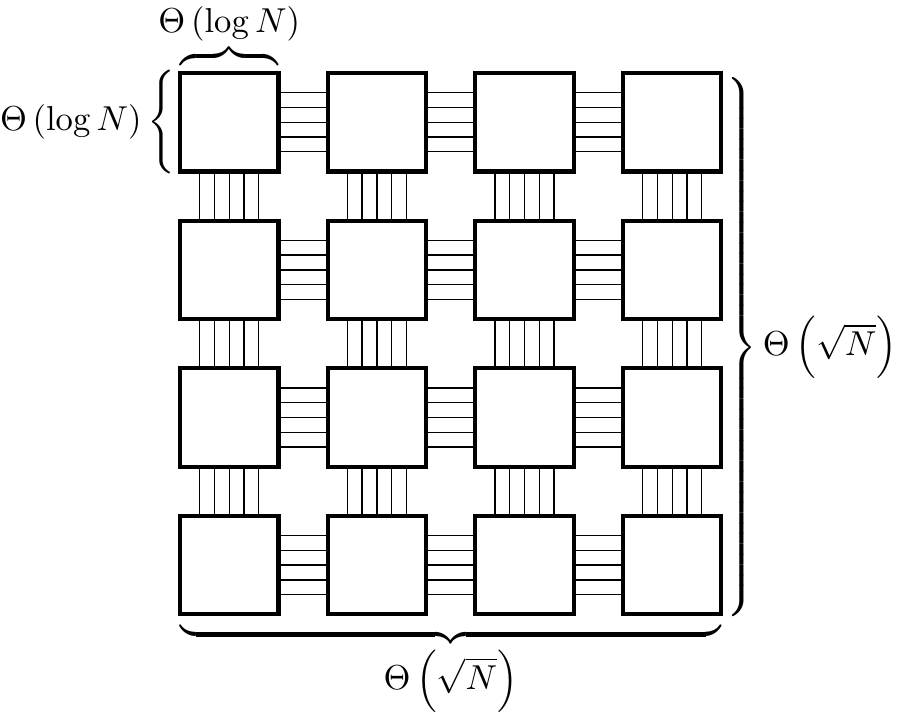}
	\caption{Diagram of a mesh network. A mesh network consists of a grid of $\sqrt N \times \sqrt N$ processor nodes, each with area that scales as $\Theta ( \log^2 N)$. Each node is connected to its at most $4$ neighbors with $\Theta (\log N)$  wires. } \label{fig:meshNetworkDiagram}
\end{figure}

A mesh network consists of a grid of processor nodes. Each processor node is capable of sending and receiving messages to
its adjacent nodes. As shown in Figure \ref{fig:meshNetworkDiagram}, each node has area that scales as $\Theta(\log^2 N)$, and consists of a processor and memory. The processor takes area that remains constant with increasing circuit size, so the amount of memory in each node can scale as $\Theta (\log^2 N)$. Each processor node must also contain instructions on what each node is to compute. The length of the instructions obviously cannot be longer than $O(\log^2 N)$. As is clear from the diagram, each processor node is connected to up to $4$ other processor nodes with $\Theta (\log N )$ wires.

A computation on a mesh network consists of a message-passing procedure and the computation that the nodes are to perform on the messages they receive.  A typical message consists of an \emph{address},
an encoding of the node to which the message is to be sent, and its
\emph{content}, the information meant to be sent. Since there are $\Theta(N)$ processor nodes, an addressing scheme with $\Theta (\log N)$ bits per address is sufficient. As well, a constant-sized message and the address of the node to which the message is to be sent can be passed between adjacent nodes in a single clock cycle, because the width of the wires connecting such nodes also scales as $\Theta(\log N)$. (Actually, messages with size that scales logarithmically in $N$ can be sent).

Multiple messages may be sent simultaneously in a mesh network, but in order for an algorithm to be \emph{valid} for a mesh-network, it
must be that no computational node is required to store more than
$O(\log^{2}N)$ bits in its memory. As well, for an algorithm to be valid, it must also avoid message-passing \emph{conflicts:}
two messages cannot be passed to the same node at the same time. Thus,
given a mesh-network algorithm, we must show that its message passing
order does not result in any conflicts.

In the section below we provide a message-passing procedure for computing
polar encoding.

\subsection{Encoding}

\begin{figure} 
	\centering
	\includegraphics{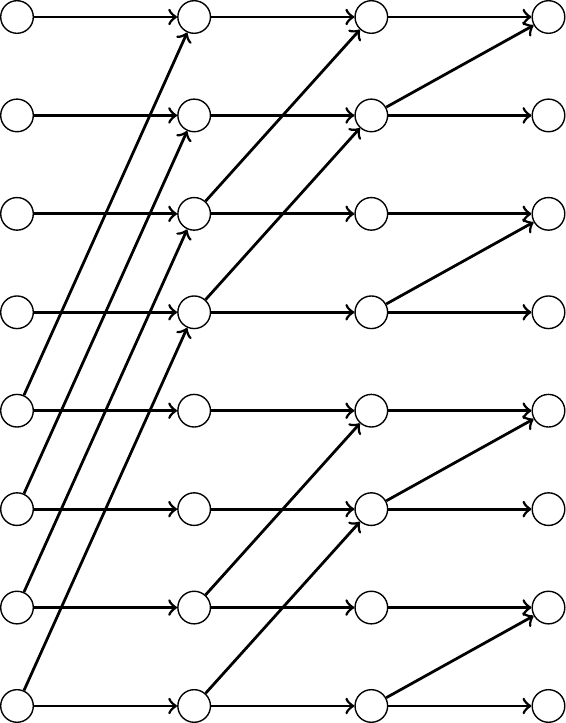}
	\caption{ Example of an encoding graph for $N=2^{3}=8$ taken directly
		from Arikan \cite{ArikanPolar}. The leftmost column of nodes are the input bits, sending their bits to to their adjacent nodes. Upon receiving these bits, the nodes in the second column compute the $\mod 2$ sum of their inputs, and then pass this result to their adjacent nodes. This procedure naturally suits
		implementation on a mesh network. Such encoding reaches the energy complexity
		lower bounds in polar encoding up to polylogarithmic factors.}
	\label{fig:encodingGraph}
\end{figure}

Arikan provides a method for computing polar encoding that naturally
lends itself to implementation on a mesh network. See Figure~\ref{fig:decoderGraph} for
a graphical representation of the Arikan method. In the Arikan method, the input nodes are on the left side of the graph. As well, for polar encoding, some of these nodes represent frozen bits. In the encoding algorithm, messages move left to right. The input nodes (and frozen bit nodes) first pass their bits to the node to which they are connected in the adjacent column of nodes. The nodes in this column proceed to compute the modulo $2$ sum of their inputs, and pass the result to their adjacent nodes on the right. This continues until the final column is reached, resulting in the codeword.

In our proposed mesh-network implementation, each of the $N$ processor nodes corresponds to a row of nodes in the encoding graph of Figure~\ref{fig:encodingGraph}. Obviously, if each message (which corresponds to an edge in the graph) is to be sent one-by one, such an order of the message-passing procedure would avoid conflicts and would be easily be implementable on a mesh network. (In fact, this is the way we propose to do decoding). However, much of the computation for encoding can be done in parallel.  We show in Appendix~\ref{app:procedudeAvoidsConflicts} how a constant fraction of the messages corresponding to edges connecting nodes in adjacent columns can be sent simultaneously in a way that avoids conflicts. We suspect that most sufficiently sparse communication graphs can be implemented on a mesh network in a way that avoids conflicts, however a general analysis of this is beyond the scope of this paper.

\subsection{Analysis of Mesh Network Encoding Algorithm Complexity}

Note that there at $\Theta(\log_{2}N)$ stages. Suppose the number
of clock cycles used by an individual node for reading an address and computing which direction to send its message is $t_{R}$. Suppose that
the complexity for computing parity of the received bit with the current
bit is $t_{P}$. At each stage, the number of hops between processor nodes is at most $O(\sqrt{N}$). There are $\Theta (\log N)$ stages.
Thus, the number of clock cycles required is
\[
T=\Theta\left(\log N\left(\sqrt{N}T_{\mathrm{R}}+T_{\mathrm{P}}\right)\right)
\]
The computation of parity can obviously be done in time $\Theta(1)$. The
routing requires computing which direction to ``send'' the message:
up, right, or left. This can easily be accomplished in $\Theta\left(\log N\right)$
time (that is, proportional to the length of the address).

The proposed algorithm also uses roughly the same fraction of node
each clock cycle, so we can assume for scaling rules the switching
activity factor ($q$) is constant. The area of such a circuit scales as $A=\Theta(\log^2 N)$.  Thus:

\[
E=qAT=\Theta(N^{3/2}\log^{4}N)
\]

\subsection{Decoding Mesh Network\label{sub:Decoding-Mesh-Network}}

Clearly, because of the requirement of successively computing each
estimate in polar decoding, asymptotically the number of clock cycles
for a polar decoding scheme must scale at least as $\Omega(N)$. A
fully parallel polar decoder thus must have area-time complexity at
least $\Omega(N^{2})$.

However, the algorithm proposed by Arikan \cite{ArikanPolar} that
takes time complexity $O(N\log N)$ can easily be implemented on a
mesh network. Each node of the mesh network corresponds to a row of
nodes of the graph in Figure~\ref{fig:decoderGraph}. As described by Arikan \cite{ArikanPolar}, a depth first message-passing procedure
between the nodes of the graph can compute the polar decoding in Turing time complexity
$O(N\log N)$. The distance between any two nodes in a mesh network
is not greater than $O(N^{1/2})$. Thus, decoding on a mesh network
takes $A=\Theta\left(N\log^{2}N\right)$ and $T=\Theta(N^{3/2}\log^{2}N)$,
where the algorithm takes $O(N\log N)$ steps, and $O(N^{1/2}\log N)$
time to do the message passing. Since a fully
parallel decoding algorithm requires only a single processing node
to be active at a given time, the switching activity factor of this
scheme scales as $\Theta\left(1/N\right)$. Thus the energy of the
computation scales as $E=O\left(N^{3/2}\log^{4}N\right)$.

\section{Generalized Polar Coding on a Mesh Network}\label{sec:generalizedOnMesh}
Arikan \cite{ArikanPolar} proposes a generalization of polar codes in which the generator matrix is no longer $G_1$ as defined in Section~\ref{sec:Polar-Encoders-Lower}. Hassani \emph{et al.} \cite{HassaniAlishahiUrbanke} analyze such schemes and show that there exist generating matrices in which for sufficiently large $N$, $P_\mathrm{N}\le e^{-n^{1-\epsilon}}$ for any $\epsilon >0$. That is, they are $e^{-n^{1-\epsilon}}$-coding schemes. By \cite[Theorem 1]{GroverInfoFriction}, such circuits must have bit-meters energy that scales as $E\ge \Omega (N^{3/2-\epsilon/2})$.  Both \cite{ArikanPolar} and \cite{HassaniAlishahiUrbanke} argue that such schemes have $O(N\log N)$ time complexity algorithms for decoding. When implemented on a mesh network, such algorithms would take energy complexity $\Theta(N^{3/2}\log^4 N)$ for the same reasons described in Section~\ref{sub:Decoding-Mesh-Network}. Thus, we can say that the general lower bounds can be almost reached for such close-to-exponential coding schemes (that is, within a factor of $N^\epsilon \mathrm{polylog}(N)$). Such a scheme would have a switching activity factor $q$ that scales as $\Theta (1/N)$. Such a circuit would take number of clock cycles that scales as $T=\Theta (N^{3/2}\log^2 N)$. Note however that the Thompson complexity analysis of \cite{blakeKschischangEnergyLatencyReliabilityTradeoffs} suggests a lower bound $T(N)$ of $\Omega (N^{1/2})$, and thus this method does not simultaneously reach these energy and time lower bounds.

We conjecture that the Thompson model lower bounds for time and energy of \cite{blakeKschischangEnergyLatencyReliabilityTradeoffs} for encoding can be reached with polar encoding up to an $N^\epsilon \mathrm {polylog}(N)$ factor for small $\epsilon$. This is because parallelization of the encoding procedure is possible. However, such a parallel algorithm must avoid conflicts on the mesh network, so proving this remains an open question.

\section{Energy Scaling as Function of gap to Capacity\label{sec:Energy-Scaling-as-function-gap-to-capacity}}
In this section, we consider how the energy of polar codes scales as
capacity is approached.\begin{defn}For a particular code, let $\chi=\frac{1}{1-\frac{R}{C}}$ be the \emph{reciprocal gap to capacity}. \end{defn}
Note that as rate approaches capacity, $\chi$ approaches infinity.

Guruswami \emph{et al.} \cite{GuruswamiXia} show that as a function
of reciprocal gap to capacity, the block length required
to achieve a set probability of block error $P_{\mathrm{e}}$ for
polar codes scale as $N=O\left(\chi^{\mu}\right)$
for some value $\mu$; that is, the block length scales polynomially
in the reciprocal gap to capacity. Hassani \emph{et al.} \cite{HassaniAlishahiUrbanke}
show that $3.55\le\mu\le6$. Goldin \emph{et al. }\cite{GoldinBurshtein}
improve the upper bound on $\mu$ to $5.7$ and Mondelli \emph{et
al. }\cite{MondelliHassaniUrbanke} further improve the upper bound
to $4.7$. These bounds, combined with Theorems~\ref{encodingLowerBoundTheorem} and \ref{decoderLowerBoundTheorem} and the discussion in Section~\ref{sec:Upper-Bounds-mesh-network} that bound energy of encoding and decoding in terms of $N$ by:
\[
\Omega\left(N^{1.5}\right)\le E_{\mathrm{comp}}\le O\left(N^{1.5}\log^{4}N\right)
\]
imply an obvious corollary.
\begin{cor}
	The energy for polar encoders with reciprocal gap to capacity $\chi$ and
	a set probability of block error, in which $C>\frac{1}{2}$, is bounded
	by:
	\begin{equation}
	\Omega\left(\chi^{5.325}\right)\le E_{\mathrm{comp}}\le O\left(\chi^{7.05}\log^{4}\left(\chi\right)\right)\label{eq:fractionOfCapacityLowerBounds}
	\end{equation}
	with decoding energy bounded similarly. 
\end{cor}
 Note that polar codes are $e^{-N^{\frac{1}{2}-\epsilon}}$-codes
\cite{arikanTeletarRateofPolarization} for any $\epsilon>0$. We can apply
the well known general lower bound on block length of any code as a
function of fraction of capacity \cite{Strassen} of $N\ge\Omega
(\chi^2)$ and the general lower bound of
\cite[Theorem 1]{blakeKschischangEnergyLatencyReliabilityTradeoffs} to show that all $e^{-N^{\frac{1}{2}-\epsilon}}$-encoding and decoding schemes have energy bounded by:
\[
E\ge\Omega\left(N^{\frac{5}{4}-\epsilon}\right)\ge\chi^{2.5}
\]
When contrasted with the lower bounds of (\ref{eq:fractionOfCapacityLowerBounds}), this
illuminates a gap between general lower bounds and that which is achievable
through polar coding.

\section{Future Work\label{sec:Future-Work}}

Extending the lower bound results to polar codes with rates less than
$1/2$ for encoder lower bounds and less than $2/3$ for decoding lower bounds is an obvious area of future work. Moreover, we have not actually implemented the mesh network topology analyzed in Section \ref{sec:Upper-Bounds-mesh-network}, so this also remains an area of future work. 

Our decoding lower bounds are for algorithms based on graphs obtained from the butterfly network decoding graph suggested by Arikan \cite{ArikanPolar}. Our lower bound does not necessarily include all decoders based on successive cancellation decoding for polar codes. A particular challenge for generalizing this result is in defining precisely what a circuit that performs a polar decoding algorithm actually is, and thus this remains an area of future work.

Currently, our results apply to polar codes as defined in \cite{ArikanPolar}. Arikan
also suggested that other matrices can be used to generate polar codes
with similar decoding and encoding complexity. Korada \emph{et al.
} \cite{KoradoSasogluUrbanke} showed that there exist matrices which
generate polar codes that have block error probability close to $O(2^{-N})$.
Our work in \cite{blakeKschischangEnergyLatencyReliabilityTradeoffs} shows that such codes necessarily have encoding and
decoding energy that scales as $\Omega\left(n^{1.5}\right)$, and
implementation over a mesh network as described in Section~\ref{sec:Upper-Bounds-mesh-network} allows
for this encoding and decoding in a way that scales as $O(n^{1.5}\mathrm{polylog}\left(n\right)$).
Thus, this means polar codes are energy-optimal over all exponentially
low probability of error decoders.

However, a full understanding of VLSI complexity of polar encoding
should relate the generator matrix to their asymptotic VLSI complexity,
including energy, time, and area tradeoffs. Our lower bound approach
here works because of the structure of the $G_{1}$ matrix, allowing
us to derive Lemma \ref{lem:LemmaForRecursion}. This lemma is used
in a key step in the inductive proof showing that polar encoding has
high VLSI complexity. However, this approach does not seem to extend to different
generator matrices; thus we cannot rule out that there are some polarizing
matrices with encoding energy that scales better than $\Omega(n^{1.5})$.
These may possibly come at the cost of higher error probability. For
engineering purposes, however, this may be a worthwhile tradeoff,
and the analysis of polar coding may be a tool to characterize achievable
coding schemes that reach different points in this fundamental tradeoff.

\section{Appendix 1}

\appendices{}

\renewcommand{\thesubsection}{\Alph{subsection}}

\section{}\label{app:uniformOutputRectangleMBWLemma}

\begin{IEEEproof}
	(Of Lemma~\ref{lem:rectangleMBWEnergyLemma})
	Note that across any bisection
	of the communication graph of a given circuit, there are at least $\omega$
	edges that cross that bisection, with total length at least $\omega$.
	We can construct a single bisection with a line through the middle
	of the nodes. Then we can construct another bisection by shifting
	this line left one unit, and sweeping in a parallel line from the
	right. Such a bisection has half the nodes in between the two lines and half the nodes outside the two lines.
	We can then shift these two lines left one unit again. We can
	do this on the order of $\sqrt{N}$ times and each time the edges
	crossing the bisecting lines must be at least $\omega$, and each
	time the bisecting lines are in a different position, thus for each
	bisection the parts of the lines that cross the bisecting lines are
	different. In total there must be at least $\Omega\left(\sqrt{N}\omega\right)$
	total distance of the lines and thus this amount of total energy.
\end{IEEEproof}

\section{Mesh Network Encoding Procedure}\label{app:procedudeAvoidsConflicts}

\begin{figure} 
	\centering
	\includegraphics{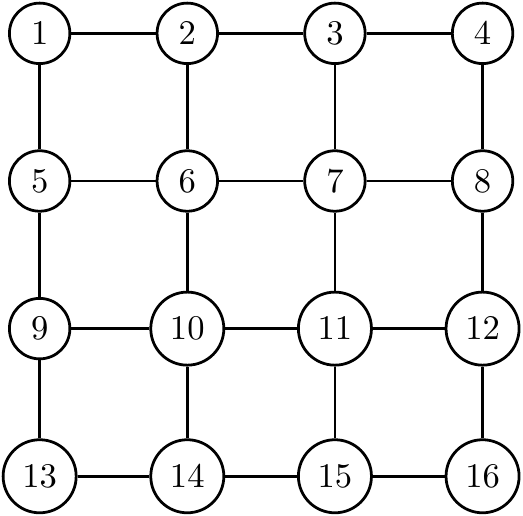}
	\caption{The "raster-scan ordering" proposed for the nodes $i=1$ to $16$ for an $N=16$ polar
		encoding circuit. Note that for general $N$ these may not fit perfectly
		on a square. In that case, we propose using dimensions that are as
		close to square-like as possible. In Appendix~\ref{app:procedudeAvoidsConflicts} we show that labelling nodes like this and applying our proposed message-passing procedure for polar encoding avoids conflicts and thus is valid.}
	\label{fig:meshNodeOrdering}
\end{figure}

We implement the message-passing procedure on a mesh network using $n=\log N$ stages, one stage for each pair of adjacent columns in the encoding graph of Figure~\ref{fig:encodingGraph}. Label the input nodes of the encoding graph in order, starting from the top, from $1$ to $N$. For each input node $i$ of the encoding graph, associate a mesh network processor node $i$. This processor node is to perform all the computations and message passings of the graph nodes in the same row as its associated input node. Place the processor nodes on the $\sqrt N \times \sqrt N$ mesh network in the order shown in Figure.~\ref{fig:meshNodeOrdering}. We call such an ordering a \emph{raster scan ordering}. Note that by inspecting Figure~\ref{fig:encodingGraph}, in the $j$th stage of the procedure, each processor node $i$ must pass messages to node $i-2^{n-j}$. 

There is some ambiguity in our definition of the operation of the mesh network: when a message is received by a processor node, in which direction should the message be sent? It is clear that if the message is located "above and to the left" of the node, the message should be passed either to the left or above. We shall adopt the convention that a node shall choose the relevant up or down direction before deciding to send the node left or right, which occurs at what we will call the \emph{target row} (that is, the row containing the computational node to which the message is sent).

We can label each row of the mesh network, and thus some  nodes will be on even rows, and some nodes will be on odd rows. In our proposed procedure, each stage of the encoding will be divided into two message-passing steps: the "even-row" passing step and the "odd-row" passing step. More precisely, at the $j$th stage of the encoding procedure, only  nodes $i$ on even rows that are to send their bits  to node $i-2^{n-j-1}$ shall do so. Then, the appropriate nodes on odd rows are to do the same. We claim this procedure avoid conflicts (that is, no two messages will be sent to the same node simultaneously).

\begin{defn} A \emph{constant send-back procedure} is a message passing procedure defined on a mesh network with nodes labelled according to the raster-scan ordering in which a set of nodes, indexed $i$, each simultaneously send a message to node $i-m$, for some $m>0$.
	\end{defn}
The "even row" sending step and the "odd row" sending step of the procedure we propose for polar encoding is obviously a special case of this procedure.

\begin{lem} \label{lem:sameRowNoConflicts}
	In any constant send-back procedure, conflicts can only occur with messages originating on different rows.\end{lem}
\begin{IEEEproof}
	Consider two messages originating on the same row. Since our convention is to send nodes "up" until deciding to send them left or right, a conflict between these two messages can only occur on the target row in which one message is sent left, and the other sent right. However, because of the ordering of the processor nodes, and the fact that we are considering a constant send-back procedure, upon reaching the target row, these messages must be sent in the same direction, otherwise they are not addressed to nodes a constant value less than their index. If these nodes do not have the same target rows, then the lemma flows trivially.
\end{IEEEproof}
\begin{lem} \label{lem:2RowsApart}
	Messages originating on rows spaced at least $2$ rows apart cannot
	have the same target rows and can not conflict.\end{lem}
\begin{IEEEproof} 
	Let $x$ be the number of processor nodes in each row.
	Clearly, the spacing between two nodes at least $2$ rows apart is at least $x+1$ (occurring
	when the lesser indexed node is at the far right, and the greater indexed
	node is on the far left). Suppose their target node was on the same
	row. The spacing between these two target indices must at least be $x+1$,
	but there are only $x$ elements on each row. Those, messages on rows spaced two or more apart cannot conflict in their target rows. It may be possible for them to conflict where one message has reached a target row, and is travelling left or right, and another is travelling up. However, the message travelling up must have originated on a row below the left or right-going node. In a constant send-back procedure such a message cannot have a target row at the same level or higher than the other node, so this cannot occur.  
\end{IEEEproof}

Since in each step, at each stage of our proposed procedure, no two simultaneously sent messages originating on adjacent rows, combining Lemmas~\ref{lem:sameRowNoConflicts} and \ref{lem:2RowsApart} confirms that our proposed procedure has no conflicts.

This particular message passing
ordering is done entirely to prove that the area-time complexity \emph{scaling} of
this mesh scheme is close to the lower bound. It is likely
that in any practical implementation a more efficient message passing
procedure exists (though it will likely be more efficient only up to some
constant factor).

\bibliographystyle{IEEEtran}
\bibliography{bibtextDoc}

\end{document}